\documentclass{ocg}

\usepackage{graphicx}

  \newcommand{\qmod}[1]{{\textup{[}{#1}\textup{]}}}

  \newcommand{\vs}{{\textvisiblespace}}

\begin{document}

\title{CELLULAR AUTOMATA\\GET THEIR WIRES CROSSED}

\author{Ed Blakey}
\address{Department of Mathematics, University of Bristol, University Walk, Bristol, BS8~1TW, UK\\
  \email{ed.blakey@queens.oxon.org}}

\maketitle

\begin{abstract}
In three spatial dimensions, communication channels are free to pass over or under each other so as to cross without intersecting; in two dimensions, assuming channels of strictly positive thickness, this is not the case.
It is natural, then, to ask whether one can, in a suitable, two-dimensional model, cross two channels in such a way that each successfully conveys its data, in particular without the channels interfering at the intersection.
We formalize this question by modelling channels as cellular automata, and answer it affirmatively by exhibiting systems whereby channels are crossed without compromising capacity.
We consider the efficiency (in various senses) of these systems, and mention potential applications.

\end{abstract}

\section{Introduction}\label{sec:int}

\subsection{Motivation}

Suppose that one wishes to implement two channels of communication (e.g.,\ electrical wires that carry encoded bit-streams), one between points $N$ and $S$, the other between $E$ and $W$.
As the reader may have guessed, $N$ is positioned to the north, $E$ to the east, and so on; consequently, given the desire that neither channel skirt around the other (rather, we should like shorter, more direct lines of communication), the channels must \emph{cross}.

In three spatial dimensions, this is unproblematic: one has available not only `north/south' and `east/west' but also `up/down', which confers headroom enough for one channel to pass over the other, carried by a bridge for example.
In two dimensions, however, the channels not merely cross but moreover \emph{intersect}; and it is not at all clear that the channels' respective signals can survive this intersection (if the channels are electrical wires, for example, then their intersection implies electrical \emph{connection}, whence it is no longer possible to determine on which channel a current originated).

Thus, one naturally asks: \emph{is it possible, in a suitable, two-dimensional model, to cross two channels in such a way that each successfully conveys its data, in particular without problematic interference at the intersection?}

It is upon this question, which we call the \emph{Cross Question} (\emph{CQ}), that we focus in the present paper.
Our answer is affirmative, and, moreover, our proof constructive: we exhibit cellular automata (for it is these that we adopt as our formalism) that successfully cross channels without impairing capacity.
We consider also the efficiency (according to several measures) of our automata.

As a historical note, we point out that this work was originally inspired by the low-level description of neural activity given in \cite{geb}, in which a figure (reproduced here as Fig.~\ref{fig:geb}) appears with the caption,
``\qmod{i}n this schematic diagram, neurons are imagined as laid out as dots in one plane. Two overlapping pathways are shown in different shades of gray. It may happen that two independent ``neural flashes'' simultaneously race down these two pathways, passing through one another like two ripples on a pond's surface\ldots''

\begin{figure}[htbp]
  \centering
  \includegraphics[height=45mm]{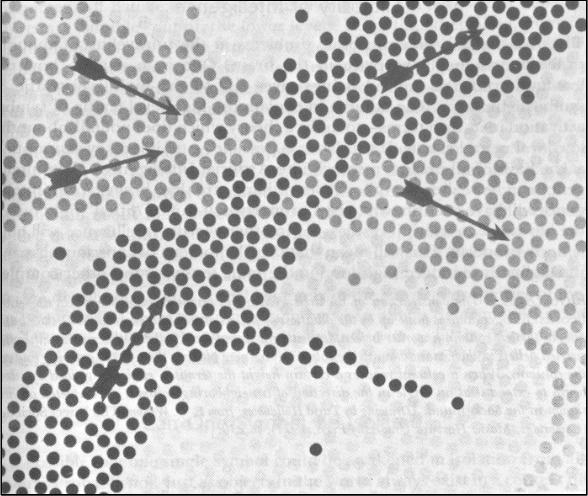}
  \caption{Two intersecting neural pathways, along which signals may travel, passing through each other at the common junction (figure reproduced from \cite{geb}).}\label{fig:geb}
\end{figure}

No further consideration is given in \cite{geb} to the details of how these signals pass through each other.
Neither should further explanation be expected: the restriction to two dimensions is made purely to allow the depiction on paper of a (three-dimensional) phenomenon in which neurons are in fact free to bend around each other and so to cross without intersecting.
Nonetheless, Fig.~\ref{fig:geb} inspires the present work by prompting the (intrinsically two-dimensional) CQ, which we answer below.
One may wonder, given that real-life neurons have available to them a third dimension, whether there is any need to consider the CQ; we claim that there is: aside from the question's academic interest, we note in Sect.~\ref{sec:conapp}\ practical contexts in which restriction to two dimensions is natural and beneficial.

\subsection{Approach}\label{sec:intapp}

So as to formalize the ideas above, and in particular so as to be able to state more rigorously the CQ, we model communication channels as \emph{cellular automata}; more precisely, we take as our model cellular automata augmented with the ability to accept \emph{input} (specifically, the messages to be carried by the channels so modelled)---cf.\ \cite{kutw}.
We give now the relevant definitions.

\begin{definition}[cf.\ \cite{canotes}]
A \emph{cellular automaton} is a tuple $\left(L, S, N, f\right)$ satisfying the following.
\begin{btlists}
  \item $L$ is a regular lattice of \emph{cells}.
  \item $S$ is a finite set of \emph{states}.
  \item $N = \left(n_1, \ldots, n_\nu\right)$ is a tuple\footnote{The ordering $n_1$, \ldots,\ $n_\nu$ is arbitrary, but is nonetheless fixed for use in Definition~\ref{def:caconf}.} of finitely many \emph{neighbourhood offsets} (which are distinct), where, for each offset $n_i$ and each cell $c \in L$, $n_i + c \in L$ is a \emph{neighbour} of $c$.
      $N$ contains amongst its coordinates the additive identity of $L$, whence each cell has itself as a neighbour.
  \item $f \colon S^\nu \to S$ is a \emph{transition function}.
\end{btlists}
\end{definition}

\begin{definition}[cf.\ \cite{canotes}]\label{def:caconf}
Suppose that $L$, $S$, $N$ and $f$ are as above.
A \emph{configuration} is a function $C_t \colon L \to S$ assigning to each cell a state; the subscript $t \in \mathbb{N}$ indexes (discretely modelled) time.
The initial configuration $C_0$ must be specified as part of the automaton's description; thereafter, each configuration $C_{t+1}$ is determined by the last, $C_t$, by way of the transition function $f$: for each cell $c \in L$ and each time $t \in \mathbb{N}$, $C_{t+1}\left(c\right) := f\left(C_t\left(c + n_1\right), \ldots, C_t\left(c + n_\nu\right)\right)$ (recall the fixed order $n_1$, \ldots,\ $n_\nu$).
\end{definition}

In modelling communication channels as cellular automata, it is convenient for us to modify the above definitions as follows.
\begin{btlisti}
  \item\label{pt:modsubset} We allow the cells to form a \emph{subset} of a lattice, rather than necessarily the whole lattice.
  \item\label{pt:modvary} We allow the transition function to \emph{vary} from cell to cell and from time-step to time-step.
  \item\label{pt:modinput} Crucially, we allow an automaton to accept \emph{input} (the messages to be carried by the channels); to this end, we equip the automaton with \emph{sources}: cells not governed by transition functions, but rather to which are supplied messages.\footnote{Compare the role of \emph{sources} here with the means by which input is supplied to \emph{iterative arrays}---see \cite{kutw}.}
\end{btlisti}

Consequently, we have the following.

\begin{definition}
A \emph{cellular automaton with sources} (\emph{CAS}) is a tuple $\mathcal{C} = \left(L, A, S, N, \left\{f_c\right\}_{c \in L \setminus A}\right)$ satisfying the following.
\begin{btlists}
  \item $L$ is a subset of some lattice ($\bar{L}$, say); elements of $L$ are \emph{cells} (cf.\ modification~(\ref{pt:modsubset}) above).\footnote{\label{fn:lfin}It is often desirable to impose the condition that $L$ be finite, though this stipulation is not necessary in order that the claims of the present paper hold (of course, the claims still hold when such restriction is made).}
  \item $A \subseteq L$ is the set of \emph{sources} of $\mathcal{C}$ (cf.\ modification~(\ref{pt:modinput}) above).
  \item $S$ is a finite set of \emph{states}, and contains a distinguished element \emph{blank}, denoted `\vs'.
      Messages to be carried by our channels are composed of states in $S$; let us suppose for convenience when transcribing messages, then, that $\left\{0, \ldots, 9\right\} \cup \left\{\mathrm{A}, \ldots, \mathrm{Z} \right\} \subseteq S$.
  \item $N = \left(n_1, \ldots, n_\nu\right)$ is a tuple of finitely many \emph{neighbourhood offsets} (which are distinct), where, for each offset $n_i$ and each cell $c \in L$, $n_i + c \in \bar{L}$ is a \emph{lattice-neighbour} of $c$; if, furthermore, $n_i + c \in L$, then $n_i + c$ is a \emph{neighbour} of $c$.
  \item For each non-source cell $c \in L \setminus A$, $f_c \colon S^\nu \times \mathbb{N} \to S$ is the \emph{transition function} for $c$.
      `$\mathbb{N}$' represents time, thus allowing the transition function to vary not only from cell to cell (hence the subscript `$c$') but also from time-step to time-step (cf.\ modification~(\ref{pt:modvary}) above).\footnote{\label{fn:parity}Note that, for present purposes, we require only (finitely representable) transition functions $f_c \colon S^\nu \times \left\{0, 1\right\} \to S$ that depend upon the \emph{parity} (an element of $\left\{0, 1\right\}$), rather than the value (an element of $\mathbb{N}$), of time; thus we avoid prohibitively (or even infinitely) complicated and memory-hungry descriptions of transition functions. Further, recall from Footnote~\ref{fn:lfin} that one may suppose $L$ to be finite, whence he need consider only finitely many distinct transition functions.}
\end{btlists}
\end{definition}

\begin{definition}
Let $\mathcal{C} = \left(L, A, S, N, \left\{f_c\right\}_{c \in L \setminus A}\right)$ be a CAS.
\begin{btlists}
  \item A \emph{message} (in $S$) is a function $m \colon \left\{0, 1, 2, \ldots, r\right\} \to S \setminus \left\{\text{\vs}\right\}$; $r \in \mathbb{N}$ is the message's \emph{length} (denoted $\lambda\left(m\right)$).
      We often describe messages in sequence form $m\left(0\right)$, \ldots,\ $m\left(\lambda\left(m\right)\right)$.
      Let $\mathcal{M}\left(S\right)$ be the set of messages (of any length) in $S$.
  \item An \emph{input} for $\mathcal{C}$ is a map $\iota \colon A \to \mathcal{M}\left(S\right)$ assigning to each source a message.
  \item	A \emph{configuration} is a function $C_{t,\iota} \colon L \to S$ assigning to each cell a state; the subscript $t \in \mathbb{N}$ discretely indexes time, whilst $\iota$ is an input.
    \begin{btlists}
      \item The initial configuration $C_{0,\iota}$ is defined such that, for all sources $c \in A$, $C_{0,\iota}\left(c\right) = \left(\iota\left(c\right)\right)\left(0\right)$, and, for all non-sources $c \in L \setminus A$, $C_{0,\iota}$ maps $c$ to an arbitrary element of $S$ (so, as with standard cellular automata, the initial configuration---in particular of the non-source cells---is a `free variable' that forms part of the system's description).
          Below, we frequently encounter the initial configuration mapping each non-source cell to the blank state \vs; we call this \emph{\vs-initialization}.
      \item Subsequent configurations $C_{t+1,\iota}$ are given by
          \[
          C_{t+1,\iota}\left(c\right) := \begin{cases} \left(\iota\left(c\right)\right)\left(t + 1\right) & \text{if } c \in A \wedge t < \lambda\circ\iota\left(c\right) \\ \text{\vs} & \text{if } c \in A \wedge t \geq \lambda\circ\iota\left(c\right) \\ f_c\left(C_{t,\iota}\left(c + n_1\right), \ldots, C_{t,\iota}\left(c + n_\nu\right), t + 1\right) & \text{otherwise} \enspace . \end{cases}
          \]
    \end{btlists}
      Hence, each source $c$ assumes in order the states $\left(\iota\left(c\right)\right)\left(0\right)$, \ldots,\ $\left(\iota\left(c\right)\right)\left(\lambda\circ\iota\left(c\right)\right)$, \vs, \vs, \vs, \ldots\ (i.e.,\ the message supplied by $\iota$ to source $c$, followed by blank states), and each non-source cell $c$ is initialized according to the choice of initial configuration and subsequently governed by its transition function $f_c$.
\end{btlists}
\end{definition}

(The formulation of transition functions given here allows, for example, the `spontaneous appearance of information', whereby non-{\vs} states may appear in neighbourhoods entirely populated by {\vs} states. Whilst this is undesirable in many contexts, it is unproblematic for present purposes, and so our definitions do not preclude such appearance.)

So as to honour modification~(\ref{pt:modsubset}) above, we intend transition function $f_c$ to take as its arguments the states only of the \emph{neighbours}---not generally all \emph{lattice}-neighbours---of $c$ (followed by a time index).
We stipulate, then, that $f_c$ not depend on its $i$th argument ($i \in \left\{1, \ldots, \nu\right\}$) whenever $c + n_i$ is not a neighbour (but merely a lattice-neighbour) of $c$; i.e.,\ for such $i$, we stipulate that, for each $\left(s_1, \ldots, s_{i-1}, s_{i+1}, \ldots, s_\nu\right) \in S^{\nu - 1}$ and each $t \in \mathbb{N}$, $\left|\left\{\, f_c\left(s_1, \ldots, s_\nu, t\right) \;\middle\vert\; s_i \in S \,\right\}\right| = 1$.

As an aside, we note that, of the three modifications introduced above to the standard definition of cellular automata, only the third (namely, provision for accepting input) admits systems that are not attainable without such modification; the other two (namely, (\ref{pt:modsubset}) partial lattices of cells, and (\ref{pt:modvary}) \mbox{time- and} cell-heterogeneity of transition functions) are mere notational conveniences.
This is because (\ref{pt:modsubset}) $S$ can be augmented by a special `non-cell' state---that invites ignorance by transition functions\mbox{---,} at which each member of $\bar{L} \setminus L$ is held; and (\ref{pt:modvary}) a family $\left\{f_c\right\}_{c \in L \setminus A}$ of time-heterogeneous transition functions can be simulated by a single time-homogeneous function $f$ provided that each cell `labels itself' (and maintains a clock)---by extending $S$ to $S \times L \times \mathbb{N}$\mbox{---,} whence $f$ ascertains which $f_c$ to simulate.\footnote{Maintenance of such clocks renders the state set infinite. However, we consider in the present paper only transition functions with finite dependencies upon time---recall Footnote~\ref{fn:parity}\mbox{---,} thus restoring state sets' finiteness.}

\begin{definition}\label{def:chan}
For $\delta \in \mathbb{Z}$,\footnote{Allowing $\delta$ to be negative is arguably unintuitive---why write `$\mathcal{C} \colon y \stackrel{-\delta}{\Rightarrow} x$' when `$\mathcal{C} \colon x \stackrel{\delta}{\Rightarrow} y$' would do?\mbox{---,}\ but has the desirable effect that the relation `between $x$ and $y$ there is a channel' is an \emph{equivalence relation}.} a CAS $\mathcal{C}$ (with cells $x$ and $y$) is said to be a \emph{channel from $x$ to $y$ with delay $\delta$}, written $\mathcal{C} \colon x \stackrel{\delta}{\Rightarrow} y$, if, for all $t \in \mathbb{N}$ and all inputs $\iota$,  $C_{t,\iota}\left(y\right) = \begin{cases} C_{t - \delta,\iota}\left(x\right) & \text{if } t \geq \delta \\ \text{\vs} & \text{otherwise} \enspace . \end{cases}$
\end{definition}

Thus, $\mathcal{C}$ is a channel from $x$ to $y$ with delay $\delta$ if and only if the states of $x$ are exactly reproduced at $y$ $\delta$ time-steps later (here, we adopt the conventions (1)~that all configurations $C_{t, \iota}$ for negative time indices $t$ are the constant blank function $C_{t, \iota} \colon c \mapsto \text{\vs}$, and (2)~that we \vs-initialize).

Note that we often take $x$ in Definition~\ref{def:chan} to be a \emph{source} of a CAS.

\begin{example}\label{eg:line}
For $q \in \mathbb{N}$, let $\mathcal{L}_q$ be the CAS $\left(\left\{0, \ldots, q\right\}, \left\{0\right\}, S, \left(-1, 0,  1\right), \left\{\pi_1\right\}\right)$, where $S$ is an arbitrary set of states and $\pi_1$ is the projection onto the first coordinate.
Then each non-source cell $c \in \left\{1, \ldots, q\right\}$ acquires as its state that of the cell $c -1$ at the previous time-step, and so any message supplied to the source $0$ propagates to a cell $c$ after $c$ time-steps.
Hence, for any $p \leq q$, $\mathcal{L}_q \colon 0 \stackrel{p}{\Rightarrow} p$.
\end{example}

Example~\ref{eg:line} captures the essence of the way in which communication channels may be expressed as cellular automata with sources: if the transition functions are such that messages' states are passed from cell to cell, then a channel may be established from a source to another cell.
Below, we consider in response to the CQ less trivial examples of CAS channels.

We focus hereafter on a specific subclass of CASs, which we now define.

\begin{definition}
A \emph{square-celled, four-neighbour CAS} (or \emph{4-CAS}) is a CAS with cells taken from $\mathbb{Z}^2$ and with neighbourhood offset tuple $N_4 := \left(\left(-1, 0\right), \left(0, -1\right), \left(1, 0\right), \left(0, 1\right), \left(0, 0\right)\right)$.
Hence, when $\mathbb{Z}^2$ is viewed as a square grid, the neighbours of a cell are the cell itself and those four cells orthogonally adjacent.
\end{definition}

\begin{definition}
\begin{btlists}
  \item Given a regular lattice of cells and a binary relation of neighbourhood between cells, a subset $C$ of the lattice, i.e.,\ a set $C$ of cells, is said to be \emph{connected} if, for every pair $\left(x, y\right)$ of cells in $C$, there exist $r \in \mathbb{N}$ and $\left(x_1, \ldots, x_r\right) \in C^r$ such that $x_1 = x$, $x_r = y$ and $x_{i+1}$ is a neighbour of $x_i$ for each $i \in \left\{1, \ldots, r - 1\right\}$.
  \item A CAS $\left(L, A, S, N, \left\{f_c\right\}_{c \in L \setminus A}\right)$ is said to be \emph{connected} if its set $L$ of cells is connected in the enveloping lattice $\bar{L}$ under the neighbourhood relation given by $N$.
\end{btlists}
\end{definition}

\begin{remark}\label{rem:bound}
Suppose that $C$ is a finite, connected set of cells (suppose also an enveloping lattice endowed with a neighbourhood relation, which, for the purposes of this remark, is viewed \emph{geometrically}; so as to ensure that the concepts used here are well defined, we assume lattice $\mathbb{Z}^2$ and neighbourhoods as for 4-CASs).
Then $C$ has an \emph{external boundary}, i.e.,\ a tuple $\left(e_0, \ldots, e_{k-1}\right)$ ($k \in \mathbb{N}$) of edges of square cells such that,
\begin{btlists}
  \item for each $i \in \left\{0, \ldots, k - 1\right\}$, $e_i$ and $e_{i + 1 \bmod k}$ meet at a vertex (say $v_i$), so that the tuple forms a \emph{closed curve} in the plane $\mathbb{R}^2$ of the geometrically-considered enveloping lattice $\mathbb{Z}^2$;
  \item the vertices $v_i$ are distinct, so that the closed curve is \emph{simple};
  \item each edge $e_i$ is the boundary between a member of $C$ (let $\bar{e}_i$ denote this unique cell), and a cell in the enveloping lattice but not in $C$; and
  \item $C$ lies entirely within the bounded region described by the curve.
\end{btlists}
It is intuitively clear that this boundary is unique up to reversal and cyclic permutation.\footnote{Existence of such a boundary can be formally demonstrated by exhibiting a set of operations on closed curves that, until a minimal-area boundary is attained, \emph{strictly reduce} the bounded (natural-number) area whilst maintaining that $C$ lies entirely within the curve, whence termination, and hence existence of a boundary, follows. Uniqueness follows from consideration of the \emph{area} enclosed by the curve.}
\end{remark}

\begin{definition}\label{def:cross}
A CAS $\mathcal{C} = \left(L, A, S, N, \left\{f_c\right\}_{c \in L \setminus A}\right)$ that is a channel both from $a_1$ to $b_1$ and from $a_2$ to $b_2$ for four distinct cells $a_1, a_2 \in A$ and $b_1, b_2 \in L$ is said to \emph{cross} these channels if
\begin{btlisti}
  \item $\mathcal{C}$ is connected;
  \item\label{pt:1touch} the external boundary $\left(e_0, \ldots, e_{k-1}\right)$ of $L$ is such that, if $\bar{e}_i = \bar{e}_j \in \left\{a_1, a_2, b_1, b_2\right\}$  (where the bar notation is as in Remark~\ref{rem:bound}, and where, without loss of generality, $i \leq j$), then either $\left|\left\{\, \bar{e}_l \;\middle\vert\; i \leq l \leq j \,\right\}\right| = 1$ or $\left|\left\{\, \bar{e}_l \;\middle\vert\; 0 \leq l \leq i \vee j \leq l \leq k - 1 \,\right\}\right| = 1$; and
  \item\label{pt:alltouch} there exist four edges $e_\alpha$, $e_\beta$, $e_\gamma$ and $e_\delta$ (with $\alpha < \beta < \gamma < \delta$) on the external boundary $\left(e_0, \ldots, e_{k-1}\right)$ of $L$ such that $\left(\bar{e}_\alpha, \bar{e}_\beta, \bar{e}_\gamma, \bar{e}_\delta\right)$ is a cyclic permutation of either $\left(a_1, a_2, b_1, b_2\right)$ or $\left(b_2, b_1, a_2, a_1\right)$.
\end{btlisti}
\end{definition}

Intuitively, point~(\ref{pt:1touch}) of Definition~\ref{def:cross} stipulates that each of $a_1$, $a_2$, $b_1$ and $b_2$ be touched at most once by the boundary of $L$ (possibly by several consecutive edges, for example when the cell is on a `corner' of $L$); point~(\ref{pt:alltouch}) stipulates that the boundary touch all of $a_1$, $a_2$, $b_1$ and $b_2$, and that it do so in such an order that $a_1$ is opposite $b_1$, $a_2$ opposite $b_2$.\footnote{Our stipulation that the boundary touch all of $a_i$ and $b_i$ is more restrictive than is necessary; if, however, there exists a method of crossing channels subject to this restriction (and we demonstrate below that such does exist), then a~fortiori there still exists such a method when the restriction is removed.}

That each $a_i$ is opposite the corresponding $b_i$ gives that, if a CAS crosses two channels as per Definition~\ref{def:cross}, then the channels do indeed cross in the intuitive sense.
Thus, in order affirmatively to answer the CQ, it is sufficient to exhibit a solution to the following problem (hereafter the \emph{4-CAS Problem}).
\begin{center}
\emph{Find a 4-CAS that crosses two channels (regardless of the choice of state set $S$)}.\end{center}

We exhibit just such a 4-CAS (in fact, two such) in Sect.~\ref{sec:sol}

\subsection{Related Work}\label{sec:intrel}

It has long been known that three XOR ($\oplus$) gates can, using only two dimensions, cross two paths carrying \emph{binary} signals; see Fig.~\ref{fig:xor}(a).
This scheme is correct since, given upper and lower input bits $a$ and $b$ respectively, the circuit produces upper and lower output bits $a \oplus \left(a \oplus b\right) = b$ and $\left(a \oplus b\right) \oplus b = a$ respectively. Geometrically speaking, as bit $a$ follows its path from (upper) input to (lower) output, it is twice combined via XOR with the same value (namely, $b$), thus leaving it unchanged; similarly, $b$ is ultimately unchanged by its twice being combined with $a$.
Furthermore, such logic circuits can be implemented in two-dimensional cellular automata such as Conway's \emph{Life} (see, for example, \cite{bcg,gar,kut}); thus, the scheme of Fig.~\ref{fig:xor}(a) solves the CQ in the case of a binary alphabet $\left\{0, 1\right\}$.

The scheme generalizes naturally to alphabets of size $n$ (w.l.o.g.\ taken to be $\left\{0, \ldots, n - 1\right\}$) via replacement of the XOR operation with addition/subtraction modulo $n$; see Fig.~\ref{fig:xor}(b).

\begin{figure}[htbp]
  \centering
  \includegraphics[height=25mm]{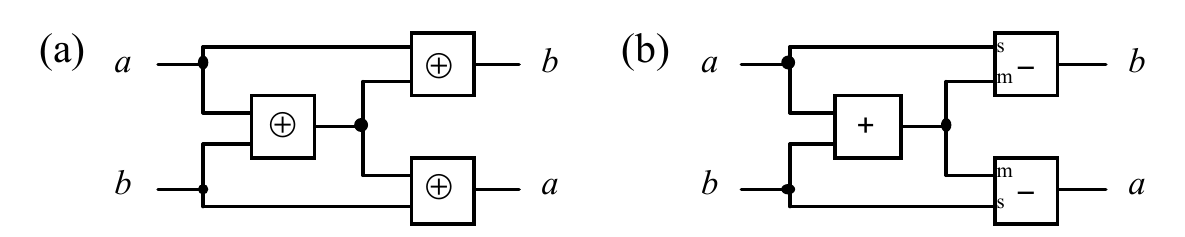}
  \caption{(a)~Scheme whereby binary values $a$ and $b$ are exchanged. This solves the CQ in the special case of messages taken from the alphabet $\left\{0, 1\right\}$. (b)~Generalization of the scheme to the alphabet $\left\{0, \ldots, n - 1\right\}$; addition and subtraction are performed modulo $n$, and `m' and `s' label each subtraction operation's minuend and subtrahend respectively.}\label{fig:xor}
\end{figure}

Since addition/subtraction modulo $n$ for fixed $n$ can clearly be implemented in cellular automata (e.g.,\ via transition functions that perform a look-up from $\mathcal{O}\left(n^2\right)$ entries), the generalized scheme of Fig.~\ref{fig:xor}(b) offers a solution to the CQ.
However, the novel solutions advocated in Sect.~\ref{sec:sol}\ of the present paper have three chief advantages over this scheme.
\begin{btlists}
  \item The solution of Fig.~\ref{fig:xor}(b) depends upon the size $n$ of the alphabet: addition and subtraction are performed \emph{modulo $n$}, whence one must have a~priori knowledge of $n$ in order to be able to implement these operations. The solutions of Sect.~\ref{sec:sol}\ are independent of the choice of alphabet; they can be implemented before $n$ is known, can accommodate the alphabet's changing during transmission of the messages to be crossed, and require of the alphabet no additive structure. 
  \item The solution of Fig.~\ref{fig:xor}(b) necessitates computational processing---e.g.,\ calculation or look-up from $\mathcal{O}\left(n^2\right)$ entries---at the addition/subtraction nodes. The solutions of Sect.~\ref{sec:sol}\ require of these nodes mere transfer (between cells) of states, which incurs no computational cost.
  \item The solution of Fig.~\ref{fig:xor}(b) preserves only the \emph{information content} of the messages being crossed, whereas the solutions of Sect.~\ref{sec:sol}\ physically move the messages' states from cell to cell. This latter, more general approach allows the crossing not only of streams of messages' symbols, but also of streams of \emph{physical objects} (whereas, of course, physical objects cannot be duplicated and combined modulo $n$ so as to be crossed by the scheme of Fig.~\ref{fig:xor}(b))---see the discussion of road junctions in Sect.~\ref{sec:conapp}
\end{btlists}

We recall also previous cellular-automatic solutions to the channel-crossing problem (see for example \cite{banks} and the references therein), but note that these do not allow the crossing of messages consisting of \emph{arbitrary} symbols, but rather reserve states that encode wires' boundaries, or utilize wires of width strictly greater than one, or similar; this is in contrast with the solutions presented in the following section.

\section{Solutions}\label{sec:sol}

In this section, we describe two solutions to the 4-CAS Problem.
The solutions both adhere to the same general scheme, which we now discuss.

\subsection{General Scheme}\label{sec:solgen}

Consider first the 4-CAS $\mathcal{X}$ given in Fig.~\ref{fig:x}.
(The lattice structure relating neighbours of $\mathcal{X}$---that $x = a_1 + \left(1, 0\right)$ and so on---is left implicit in the geometry of Fig.~\ref{fig:x}; in particular, we give the cells symbolic names, rather than labelling cells with elements of the enveloping lattice $\mathbb{Z}^2$. We treat subsequent 4-CASs similarly.)

\begin{figure}[htbp]
  \centering
  \includegraphics[height=24mm]{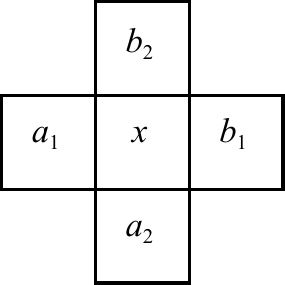}
  \caption{$\mathcal{X} = \left(\left\{a_1, a_2, x, b_1, b_2\right\}, \left\{a_1, a_2\right\}, S, N_4, \left\{f_x, f_{b_1}, f_{b_2}\right\}\right)$.}\label{fig:x}
\end{figure}

We claim that there is no choice of transition functions for the three non-source cells $x$, $b_1$ and $b_2$ such that $\mathcal{X}$ acts simultaneously as a channel from $a_i$ to $b_i$ for both $i \in \left\{1, 2\right\}$ (so certainly there is no choice such that $\mathcal{X}$ crosses these channels).
Intuitively, this is because, were these channels $\mathcal{X} \colon a_i \stackrel{\delta_i}{\Rightarrow} b_i$ established, then $x$ (by virtue of the geometry of $\mathcal{X}$) would have to encode in its state the previous states of both $a_1$ and $a_2$ (whence $b_1$ and $b_2$ could subsequently attain their respective states); however, assuming the non-trivial case where $S \supsetneq \left\{\text{\vs}\right\}$, there is no injective map from $S \times S$ to $S$, and so the (single) state of $x$---one of $\left|S\right|$ possibilities---cannot encode the \emph{pair} of previous states of $a_1$ and $a_2$---for which there are $\left|S\right|^2 > \left|S\right|$ possibilities.

Nonetheless, it is clear that either one of the channels \emph{in isolation} can be implemented via suitable choice of transition functions: if $x$ takes as its state that of $a_1$ at the previous time-step, and $b_1$ that of $x$ (more formally, if $f_x$ and $f_{b_1}$ are the projection $\pi_1$ onto the first coordinate, where, recall, the first coordinate $n_1$ of $N_4$ is $\left(-1, 0\right)$), then $\mathcal{X} \colon a_1 \stackrel{2}{\Rightarrow} b_1$; if \emph{instead} $f_x = f_{b_2} = \pi_2$, where the second coordinate $n_2$ of $N_4$ is $\left(0, -1\right)$, then $\mathcal{X} \colon a_2 \stackrel{2}{\Rightarrow} b_2$.

Further, there is a compromise between these two single channels: if $x$ takes its state \emph{alternately} from $a_1$ and $a_2$---e.g.,\ if we have the transition function $f_x \colon \left(s_1, \ldots, s_5, t\right) \mapsto \begin{cases} s_1 & \text{if $t$ is even} \\ s_2 & \text{if $t$ is odd} \end{cases}$ (and if $f_{b_i} = \pi_i$)\mbox{---,} then \emph{half} of each message (specifically every other state) is passed with delay $2$ from $a_i$ to $b_i$.
This behaviour is encapsulated in Table~\ref{tab:x}(a), which shows the state of each cell in each time-step $t \leq 6$ (arbitrarily assuming \vs-initialization, and supposing supply of messages\footnote{These are not strictly messages since it is not clear how or even if they terminate; this is not, however, problematic since their truncation at an arbitrary point yields genuine messages. Similar abuses of the definition of `message' occur throughout.} A, B, C, {\ldots}\ to $a_1$ and 0, 1, 2, {\ldots}\ to $a_2$).

\begin{table}[htbp]\footnotesize
  \begin{center}
\begin{tabular}{rcrc}
(a)&
  \begin{tabular}[t]{r|ccccc}
    $t$ & $a_1$ & $a_2$ & $x$ & $b_1$ & $b_2$ \\
    \hline
    0 & A & 0 & \vs & \vs & \vs \\
    1 & B & 1 & 0 & \vs & \vs \\
    2 & C & 2 & B & 0 & \textbf{0} \\
    3 & D & 3 & 2 & \textbf{B} & B \\
    4 & E & 4 & D & 2 & \textbf{2} \\
    5 & F & 5 & 4 & \textbf{D} & D \\
    6 & G & 6 & F & 4 & \textbf{4} \\
    \vdots & \vdots & \vdots & \vdots & \vdots & \vdots
  \end{tabular}
&(b)&
  \begin{tabular}[t]{r|ccccc}
    $t$ & $a_1$ & $a_2$ & $x$ & $b_1$ & $b_2$ \\
    \hline
    0 & A & 0 & \vs & \vs & \vs \\
    1 & B & 1 & 0 & \vs & \vs \\
    2 & C & 2 & B & \vs & \textbf{0} \\
    3 & D & 3 & 2 & \textbf{B} & \textbf{0} \\
    4 & E & 4 & D & \textbf{B} & \textbf{2} \\
    5 & F & 5 & 4 & \textbf{D} & \textbf{2} \\
    6 & G & 6 & F & \textbf{D} & \textbf{4} \\
    \vdots & \vdots & \vdots & \vdots & \vdots & \vdots
  \end{tabular}
\end{tabular}
  \caption{The behaviour of $\mathcal{X}$, (a)~as originally introduced and (b)~as modified so as to emphasize the roles of $b_i$.}\label{tab:x}
  \end{center}
\end{table}

So as to emphasize the intended roles of $b_i$---i.e.,\ that they are the respective destinations of messages supplied to $a_i$\mbox{---,} we modify $f_{b_i}$ such that $f_{b_i} \colon \left(s_1, \ldots, s_5, t\right) \mapsto \begin{cases} s_i & \text{if $t + i$ is even} \\ s_5 & \text{if $t + i$ is odd} \end{cases}$ (recall that the fifth coordinate $n_5$ of $N_4$ is $\left(0, 0\right)$).
This results in the behaviour shown in Table~\ref{tab:x}(b), where it can be seen that only `relevant' states (i.e.,\ those from the message supplied to $a_i$) arrive at each destination $b_i$.
Hereafter, $\mathcal{X}$ denotes the 4-CAS described here, including the transition function $f_x$ and modified transition functions $f_{b_i}$.

\begin{definition}
A message $m \colon \left\{0, 1, 2, \ldots, r\right\} \to S \setminus \left\{\text{\vs}\right\}$ is said to be a \emph{couplet message} if $r$ is odd and, for each even $i \in \left\{0, 2, 4, \ldots, r - 1\right\}$, $m\left(i\right) = m\left(i + 1\right)$.
Hence, couplet messages have the form \textup{A}, \textup{A}, \textup{B}, \textup{B}, \textup{C}, \textup{C}, {\ldots},\ \textup{Z}, \textup{Z}, where \textup{A}, {\ldots},\ \textup{Z} are (non-\vs) states.
\end{definition}

\begin{lemma}\label{lem:coup}
\begin{btlisti}
  \item If the message supplied to each source $a_i$ of $\mathcal{X}$ is a couplet message, then it is transmitted to its destination $b_i$ in its entirety (despite being only `half-transmitted' as outlined above).
  \item\label{pt:divcoup} Any message can be divided into two couplet messages.
  \item Any message can be recovered from the two couplet messages of point~(\ref{pt:divcoup}).
\end{btlisti}
\end{lemma}

\begin{proof}
\begin{btlisti}
  \item Table~\ref{tab:x3} shows the behaviour of $\mathcal{X}$ when \vs-initialized and supplied with arbitrary couplet messages A, A, B, B, {\ldots}\ and 0, 0, 1, 1, {\ldots}
Note that, given such input, $\mathcal{X} \colon a_1 \stackrel{3}{\Rightarrow} b_1$ and $\mathcal{X} \colon a_2 \stackrel{2}{\Rightarrow} b_2$; the respective messages are recreated in their entirety at $b_i$ since, although half of each message's states are discarded at $x$, the remaining half contain the whole information content of the message due to the redundancy conferred by the form thereof.
\begin{table}[htbp]\footnotesize
  \begin{center}
  \begin{tabular}{r|ccccc}
    $t$ & $a_1$ & $a_2$ & $x$ & $b_1$ & $b_2$ \\
    \hline
    0 & A & 0 & \vs & \vs & \vs \\
    1 & A & 0 & 0 & \vs & \vs \\
    2 & B & 1 & A & \vs & \textbf{0} \\
    3 & B & 1 & 1 & \textbf{A} & \textbf{0} \\
    4 & C & 2 & B & \textbf{A} & \textbf{1} \\
    5 & C & 2 & 2 & \textbf{B} & \textbf{1} \\
    6 & D & 3 & C & \textbf{B} & \textbf{2} \\
    \vdots & \vdots & \vdots & \vdots & \vdots & \vdots
  \end{tabular}
  \caption{The behaviour of $\mathcal{X}$ when supplied with couplet messages.}\label{tab:x3}
  \end{center}
\end{table}
  \item Consider the 4-CAS $\mathcal{Y}$ (given in Fig.~\ref{fig:y}), with transition functions $f_c \colon \left(s_1, \ldots, s_5, t\right) \mapsto \begin{cases} s_3 & \text{if $t$ is even} \\ s_5 & \text{if $t$ is odd} \end{cases}$ and $f_d \colon \left(s_1, \ldots, s_5, t\right) \mapsto \begin{cases} s_5 & \text{if $t$ is even} \\ s_1 & \text{if $t$ is odd} \end{cases}$ (recall that the third coordinate $n_3$ of $N_4$ is $\left(1, 0\right)$).
      Supposing \vs-initialization and supply to $y$ of the message A, B, C, {\dots}, the behaviour of $\mathcal{Y}$ is as in Table~\ref{tab:y}; in particular, the message is split into two couplet messages (preceded by one or two blank states), one containing message-states taken by $y$ at odd time-steps, the other containing states of $y$ at even time-steps.
\begin{figure}[htbp]
  \centering
  \includegraphics[height=8mm]{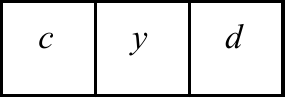}
  \caption{$\mathcal{Y} = \left(\left\{c, y, d\right\}, \left\{y\right\}, S, N_4, \left\{f_c, f_d\right\}\right)$.}\label{fig:y}
\end{figure}
\begin{table}[htbp]\footnotesize
  \begin{center}
  \begin{tabular}{r|ccc}
    $t$ & $y$ & $c$ & $d$ \\
    \hline
    0 & A & \vs & \vs \\
    1 & B & \vs & A \\
    2 & C & B & A \\
    3 & D & B & C \\
    4 & E & D & C \\
    5 & F & D & E \\
    \vdots & \vdots & \vdots & \vdots
  \end{tabular}
  \caption{The behaviour of $\mathcal{Y}$.}\label{tab:y}
  \end{center}
\end{table}
  \item Consider the 4-CAS $\mathcal{L}$ (Fig.~\ref{fig:l}), where $f_l \colon \left(s_1, \ldots, s_5, t\right) \mapsto \begin{cases} s_3 & \text{if $t$ is even} \\ s_1 & \text{if $t$ is odd} \enspace . \end{cases}$
      Supposing \vs-initialization, the behaviour of $\mathcal{L}$ given messages A, B, C, {\ldots}\ and 0, 1, 2, {\ldots}\ is shown in Table~\ref{tab:l}(a).
      If both messages are couplets (say A, A, B, B, {\ldots}\ and 0, 0, 1, 1, {\ldots}),\ then $\mathcal{L}$ exhibits the behaviour of Table~\ref{tab:l}(b).
      If, more specifically, the messages are the two couplets A, A, C, C, {\ldots}\ and B, B, D, D, {\ldots}\ produced as in point~(\ref{pt:divcoup}) by cells $d$ and $c$ of $\mathcal{Y}$, then $\mathcal{L}$ exhibits the behaviour given in Table~\ref{tab:l}(c); note in particular that the original message A, B, C, {\ldots}\ of point~(\ref{pt:divcoup}) is reconstructed at $l$ from the two couplet messages.
\end{btlisti}
\begin{figure}[htbp]
  \centering
  \includegraphics[height=8mm]{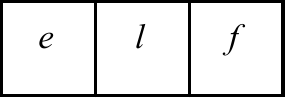}
  \caption{$\mathcal{L} = \left(\left\{e, l, f\right\}, \left\{e, f\right\}, S, N_4, \left\{f_l\right\}\right)$.}\label{fig:l}
\end{figure}
\begin{table}[htbp]\footnotesize
  \begin{center}
\begin{tabular}{rcrcrc}
(a)&
  \begin{tabular}[t]{r|ccc}
    $t$ & $e$ & $f$ & $l$ \\
    \hline
    0 & A & 0 & \vs \\
    1 & B & 1 & A \\
    2 & C & 2 & 1 \\
    3 & D & 3 & C \\
    4 & E & 4 & 3 \\
    \vdots & \vdots & \vdots & \vdots
  \end{tabular}
&(b)&
  \begin{tabular}[t]{r|ccc}
    $t$ & $e$ & $f$ & $l$ \\
    \hline
    0 & A & 0 & \vs \\
    1 & A & 0 & A \\
    2 & B & 1 & 0 \\
    3 & B & 1 & B \\
    4 & C & 2 & 1 \\
    \vdots & \vdots & \vdots & \vdots
  \end{tabular}
&(c)&
  \begin{tabular}[t]{r|ccc}
    $t$ & $e$ & $f$ & $l$ \\
    \hline
    0 & A & B & \vs \\
    1 & A & B & A \\
    2 & C & D & B \\
    3 & C & D & C \\
    4 & E & F & D \\
    \vdots & \vdots & \vdots & \vdots
  \end{tabular}
\end{tabular}
  \caption{The behaviour of $\mathcal{L}$ when supplied with (a)~arbitrary messages, (b)~couplet messages and (c)~the couplet messages constructed in point~(\ref{pt:divcoup}) of Lemma~\ref{lem:coup}.}\label{tab:l}
  \end{center}
\end{table}
\end{proof}

We see in Lemma~\ref{lem:coup} (1)~that the automaton $\mathcal{X}$ successfully crosses channels provided that the messages carried thereby are couplets, (2)~that $\mathcal{Y}$ divides an arbitrary message into two couplets, and (3)~that $\mathcal{L}$ recombines these two couplets, yielding the original message.
This suggests a general scheme whereby two channels (carrying arbitrary messages) may be crossed: the scheme uses (a)~two copies of $\mathcal{Y}$, one to divide each message into two couplets; (b)~four copies of $\mathcal{X}$, one to cross each couplet of the first message with each of the second; and (c)~two copies of $\mathcal{L}$, one to recover each message from its couplet pair.
The scheme is depicted in Fig.~\ref{fig:scheme},\footnote{Note that the arrows entering or leaving instances of the automata $\mathcal{X}$, $\mathcal{Y}$ and $\mathcal{L}$ resemble the letters `X', `Y' and `$\lambda$' respectively, hence the choice of names for these systems.} in which the messages to be crossed are $m_1$ and $m_2$; $e_i$ denotes the couplet message consisting of the even-index states of $m_i$ (that is, $e_i$ is $m_i\left(0\right)$, $m_i\left(0\right)$, $m_i\left(2\right)$, $m_i\left(2\right)$, {\ldots});\ $o_i$ is the corresponding odd-index couplet message $m_i\left(1\right)$, $m_i\left(1\right)$, $m_i\left(3\right)$, $m_i\left(3\right)$, {\ldots}

\begin{figure}[htbp]
  \centering
  \includegraphics[height=68mm]{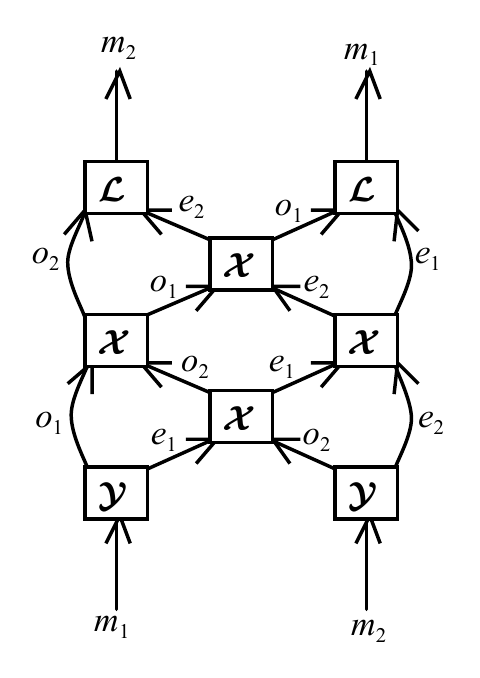}
  \caption{The scheme whereby two channels are crossed.}\label{fig:scheme}
\end{figure}

In Sect.~\ref{sec:solimp},\ we exhibit a 4-CAS that implements this scheme, thereby answering (affirmatively and constructively) the 4-CAS Problem.

\subsection{Implementing the Scheme}\label{sec:solimp}

Let $\mathcal{F}$ be the 4-CAS given in Fig.~\ref{fig:f}.
Its sources are $a_1$ and $a_2$; its neighbourhood index tuple is $N_4 = \left(\left(-1, 0\right), \left(0, -1\right), \left(1, 0\right), \left(0, 1\right), \left(0, 0\right)\right)$.

\begin{figure}[htbp]
  \centering
  \includegraphics[height=40mm]{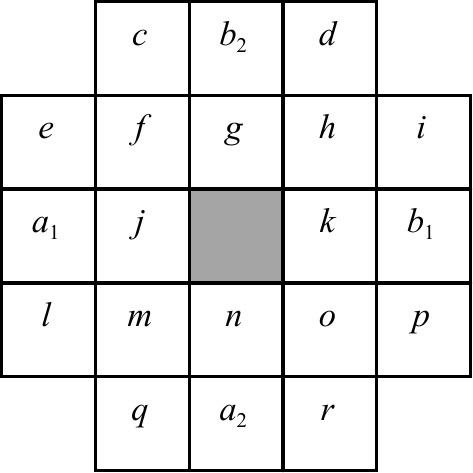}
  \caption{$\mathcal{F} = \left(\left\{a_1, a_2, b_1, b_2, c, d, e, \ldots, r\right\}, \left\{a_1, a_2\right\}, S, N_4, \left\{f_x\right\}_{x \in \left\{b_1, b_2, c, d, e, \ldots, r\right\}}\right)$.}\label{fig:f}
\end{figure}

Let `$x \leftarrow i;j$' denote `$f_x \colon \left(s_1, \ldots, s_5, t\right) \mapsto \begin{cases} s_i & \text{if $t$ is even} \enspace \text{'} \\ s_j & \text{if $t$ is odd}  \end{cases}$ and let `$X \leftarrow i;j$' denote `$\forall x \in X, \; x \leftarrow i;j$'.
Then the transition functions of $\mathcal{F}$ are as follows.
$b_1 \leftarrow 4;2$.
$b_2 \leftarrow 1;3$.
$\left\{c, d, j, k\right\} \leftarrow 5;2$.
$e \leftarrow 2;5$.
$\left\{f, h, m, o\right\} \leftarrow 2;1$.
$\left\{g, i, n, p\right\} \leftarrow 1;5$.
$l \leftarrow 5;4$.
$q \leftarrow 3;5$.
$r \leftarrow 5;1$.

Supposing that $a_1$ is supplied with the message A, B, C, {\ldots}\ and $a_2$ with 0, 1, 2, {\ldots},\ and supposing \vs-initialization, the behaviour of $\mathcal{F}$ is given in Table~\ref{tab:f}.

\begin{table}[htbp]\footnotesize
  \begin{center}
  \begin{tabular}{r||cccc|cccc|cccc|cccc|cccc}
    $t$ & $a_1$ & $a_2$ & $c$ & $d$ & $e$ & $f$ & $g$ & $h$ & $i$ & $j$ & $k$ & $l$ & $m$ & $n$ & $o$ & $p$ & $q$ & $r$ & $b_1$ & $b_2$ \\
    \hline
    \hline
    0 & A & 0 & \vs & \vs & \vs & \vs & \vs & \vs & \vs & \vs & \vs & \vs & \vs & \vs & \vs & \vs & \vs & \vs & \vs & \vs \\
    1 & B & 1 & \vs & \vs & \vs & \vs & \vs & \vs & \vs & \vs & \vs & A & \vs & \vs & \vs & \vs & \vs & 0 & \vs & \vs \\
    2 & C & 2 & \vs & \vs & B & \vs & \vs & \vs & \vs & \vs & \vs & A & \vs & \vs & 0 & \vs & 1 & 0 & \vs & \vs \\
    3 & D & 3 & \vs & \vs & B & B & \vs & \vs & \vs & \vs & 0 & C & A & \vs & \vs & \vs & 1 & 2 & \vs & \vs \\
    \hline
    4 & E & 4 & \vs & \vs & D & \vs & B & 0 & \vs & \vs & 0 & C & 1 & A & 2 & \vs & 3 & 2 & \vs & \vs \\
    5 & F & 5 & \vs & 0 & D & D & B & B & \vs & 1 & 2 & E & C & A & A & \vs & 3 & 4 & \vs & \vs \\
    6 & G & 6 & \vs & 0 & F & 1 & D & 2 & B & 1 & 2 & E & 3 & C & 4 & A & 5 & 4 & \vs & \vs \\
    7 & H & 7 & 1 & 2 & F & F & D & D & B & 3 & 4 & G & E & C & C & A & 5 & 6 & \textbf{A} & \textbf{0} \\
    \hline
    8 & I & 8 & 1 & 2 & H & 3 & F & 4 & D & 3 & 4 & G & 5 & E & 6 & C & 7 & 6 & \textbf{B} & \textbf{1} \\
    9 & J & 9 & 3 & 4 & H & H & F & F & D & 5 & 6 & I & G & E & E & C & 7 & 8 & \textbf{C} & \textbf{2} \\
    \vdots & \vdots & \vdots & \vdots & \vdots & \vdots & \vdots & \vdots & \vdots & \vdots & \vdots & \vdots & \vdots & \vdots & \vdots & \vdots & \vdots & \vdots & \vdots & \vdots & \vdots
  \end{tabular}
  \caption{The behaviour of $\mathcal{F}$. Notably, $\mathcal{F} \colon a_i \stackrel{7}{\Rightarrow} b_i$.}\label{tab:f}
  \end{center}
\end{table}

Note in particular that $\mathcal{F} \colon a_1 \stackrel{7}{\Rightarrow} b_1$ and $\mathcal{F} \colon a_2 \stackrel{7}{\Rightarrow} b_2$, and that, by virtue of the geometric layout of $\mathcal{F}$, \emph{these channels are crossed} (recall Definition~\ref{def:cross}).
Thus, we come to our main result.

\begin{theorem}
The 4-CAS Problem has a solution; this answers the CQ affirmatively.
\end{theorem}

\begin{proof}
$\mathcal{F}$ is a solution to the 4-CAS Problem.
\end{proof}

\begin{remark}\label{rem:schemelink}
So as to clarify the connection between $\mathcal{F}$ and the general scheme of Sect.~\ref{sec:solgen},\ we note that the neighbourhood of each of $f$, $h$, $m$ and $o$ acts as a copy of $\mathcal{X}$ (more precisely, a reflection thereof), that the sets $\left\{a_1, e, l\right\}$ and $\left\{a_2, q, r\right\}$ each act as (a rotation of) a copy of $\mathcal{Y}$, and that the sets $\left\{b_1, i, p\right\}$ and $\left\{b_2, c, d\right\}$ each act as (a rotation of) a copy of $\mathcal{L}$.
\end{remark}

\subsection{Improving the Solution}\label{sec:solopt}

We solve above the 4-CAS Problem (constructively, by exhibiting a 4-CAS $\mathcal{F}$ that crosses two channels), but thus far make no mention of the solution's \emph{efficiency}.

\begin{definition}\label{def:meas}
Let $\mathcal{C} = \left(L, A, S, N, \left\{f_c\right\}_{c \in L \setminus A}\right)$ be a CAS that crosses two channels $\mathcal{C} \colon a_i \stackrel{\delta_i}{\Rightarrow} b_i$ ($i \in \left\{1, 2\right\}$), where $\left|\left\{a_1, a_2, b_1, b_2\right\}\right| = 4$.
We define the following measures of efficiency:
\begin{btlists}
  \item the \emph{maximum delay} $\Delta\left(\mathcal{C}\right) = \max\left\{\delta_1, \delta_2\right\}$,
  \item the \emph{mean delay} $\bar{\delta}\left(\mathcal{C}\right) = \frac{\delta_1 + \delta_2}{2}$, and
  \item the automaton's \emph{size} (in cells) $\sigma\left(\mathcal{C}\right) = \left|L\right|$.
\end{btlists}
Note that $\Delta$ and $\bar{\delta}$ depend not only upon $\mathcal{C}$, but also upon the choice of channels; where this choice is not clear, a subscript of the form `$a_1, a_2 \Rightarrow b_1, b_2$' may be used, e.g.,\ `$\Delta_{a_1, a_2 \Rightarrow b_1, b_2}\left(\mathcal{C}\right)$' as a clarification of `$\Delta\left(\mathcal{C}\right)$'.
\end{definition}

\begin{remark}
We suggest that these measures are in some sense natural ones to consider, but certainly do not claim that they are unique in this respect.
Indeed, even amongst these three, the question of which measures are of most interest is situation-dependent: if, for example, the channels being implemented as CASs are used for urgent communications, then a small delay may be preferable to a small number of cells, whereas if the channels are used in some portable technology, then size and hence number of cells may trump speed and hence delay.
\end{remark}

In the case of our solution $\mathcal{F}$ to the 4-CAS Problem, $\Delta\left(\mathcal{F}\right) = \bar{\delta}\left(\mathcal{F}\right) = 7$ and $\sigma\left(\mathcal{F}\right) = 20$.

Whereas consideration of $\mathcal{F}$ is instructive due to the CAS's similarity with the scheme of Sect.~\ref{sec:solgen}\ (recall Remark~\ref{rem:schemelink}), $\mathcal{F}$ is by no means optimal with respect to the measures of Definition~\ref{def:meas}.
Concretely, we exhibit now a solution to the 4-CAS Problem preferable to $\mathcal{F}$ in each of these respects.

Let $\mathcal{G}$ be the 4-CAS given in Fig.~\ref{fig:g}.
Its sources are $a_1$ and $a_2$; its neighbourhood index tuple is $N_4 = \left(\left(-1, 0\right), \left(0, -1\right), \left(1, 0\right), \left(0, 1\right), \left(0, 0\right)\right)$.

\begin{figure}[htbp]
  \centering
  \includegraphics[height=32mm]{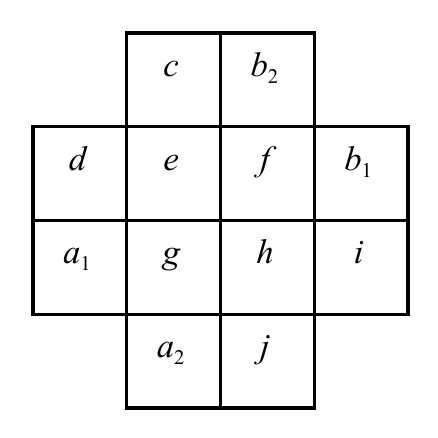}
  \caption{$\mathcal{G} = \left(\left\{a_1, a_2, b_1, b_2, c, d, e, \ldots, j\right\}, \left\{a_1, a_2\right\}, S, N_4, \left\{f_x\right\}_{x \in \left\{b_1, b_2, c, d, e, \ldots, j\right\}}\right)$.}\label{fig:g}
\end{figure}

With the `$x \leftarrow i;j$' and `$X \leftarrow i;j$' notation as in the definition of $\mathcal{F}$ above, the transition functions of $\mathcal{G}$ are as follows.
$\left\{b_1, b_2, e, h\right\} \leftarrow 1;2$.
$\left\{c, d\right\} \leftarrow 2;5$.
$\left\{f, g\right\} \leftarrow 2;1$.
$\left\{i, j\right\} \leftarrow 5;1$.

Supposing that $a_1$ is supplied with the message A, B, C, {\ldots}\ and $a_2$ with 0, 1, 2, {\ldots},\ and supposing \vs-initialization, the behaviour of $\mathcal{G}$ is given in Table~\ref{tab:g}.

\begin{table}[htbp]\footnotesize
  \begin{center}
  \begin{tabular}{r||cccc|cccc|cccc}
    $t$ & $a_1$ & $a_2$ & $c$ & $d$ & $e$ & $f$ & $g$ & $h$ & $i$ & $j$ & $b_1$ & $b_2$ \\
    \hline
    0 & A & 0 & \vs & \vs & \vs & \vs & \vs & \vs & \vs & \vs & \vs & \vs \\
    1 & B & 1 & \vs & \vs & \vs & \vs & A & \vs & \vs & 0 & \vs & \vs \\
    2 & C & 2 & \vs & B & \vs & \vs & 1 & A & \vs & 0 & \vs & \vs \\
    3 & D & 3 & \vs & B & 1 & \vs & C & 0 & A & 2 & \vs & \vs \\
    4 & E & 4 & 1 & D & B & 0 & 3 & C & A & 2 & \vs & \vs \\
    5 & F & 5 & 1 & D & 3 & B & E & 2 & C & 4 & \textbf{A} & \textbf{0} \\
    6 & G & 6 & 3 & F & D & 2 & 5 & E & C & 4 & \textbf{B} & \textbf{1} \\
    7 & H & 7 & 3 & F & 5 & D & G & 4 & E & 6 & \textbf{C} & \textbf{2} \\
    \vdots & \vdots & \vdots & \vdots & \vdots & \vdots & \vdots & \vdots & \vdots & \vdots & \vdots & \vdots & \vdots
  \end{tabular}
  \caption{The behaviour of $\mathcal{G}$. Notably, $\mathcal{G} \colon a_i \stackrel{5}{\Rightarrow} b_i$.}\label{tab:g}
  \end{center}
\end{table}

Note in particular that $\mathcal{G} \colon a_1 \stackrel{5}{\Rightarrow} b_1$ and $\mathcal{G} \colon a_2 \stackrel{5}{\Rightarrow} b_2$, and that, by virtue of the geometric layout of $\mathcal{G}$, these channels are crossed.
Further, consideration of the efficiency measures of Definition~\ref{def:meas} gives that $\Delta\left(\mathcal{G}\right) = \bar{\delta}\left(\mathcal{G}\right) = 5$ and $\sigma\left(\mathcal{G}\right) = 12$; hence, with respect to these measures, $\mathcal{G}$ offers a better solution to the 4-CAS Problem than $\mathcal{F}$.\footnote{In fact, we claim that $\mathcal{G}$ is \emph{optimal} with respect to these three measures (at least when messages' states may have \emph{physical presence} as described in the second bullet point of Sect.~\ref{sec:intrel}). We defer justification to \cite{msc}.}

Intuitively, the reason for the improvement (over $\mathcal{F}$) offered by $\mathcal{G}$ is that (a)~the sets of cells of $\mathcal{G}$ that act as copies of $\mathcal{X}$ (these sets are the respective neighbourhoods of $e$, $f$, $g$ and $h$) have greater pairwise overlaps than is the case with the corresponding sets of cells in $\mathcal{F}$; and (b)~the sets of cells of $\mathcal{G}$ acting as copies of $\mathcal{Y}$ (these sets are $\left\{a_1, d, g\right\}$ and $\left\{a_2, g, j\right\}$) or of $\mathcal{L}$ ($\left\{b_1, f, i\right\}$ and $\left\{b_2, c, f\right\}$) are better shaped---specifically, they are `L'-shaped rather than straight---to pass data to/receive data from the rest of the cells.

\section{Conclusion}\label{sec:con}

\subsection{Summary}

We consider in the present paper the problem of crossing two channels; in particular, we stipulate exactly two spatial dimensions, whence the crossed channels necessarily intersect.
We formalize this situation in Sect.~\ref{sec:intapp}\ using \emph{cellular automata}, modified so as to endow the systems with the ability to accept inputs (i.e.,\ the messages to be carried by the channels).
In Sect.~\ref{sec:intrel},\ we recall previous approaches to this question, but note that the solutions presented here have over these approaches the advantage (amongst others) of being able to cross streams not merely of \emph{information}, but also of \emph{physical objects}.

In Sect.~\ref{sec:solgen},\ we exhibit a scheme whereby channels may be crossed; the key is a simple sub-automaton $\mathcal{X}$ that, whilst unable to cross channels carrying arbitrary messages, is at least able successfully to cross \emph{couplet} messages; then arbitrary messages can be crossed by splitting them into couplets, crossing with $\mathcal{X}$, and recombining the couplets into the original messages.
In Sect.~\ref{sec:solimp},\ we implement this scheme as a 4-CAS $\mathcal{F}$, thus answering the motivating question of the present work: \emph{it is, in two dimensions, indeed possible to cross channels without impairing their capacity}.
In Sect.~\ref{sec:solopt},\ we consider the efficiency of $\mathcal{F}$, introducing measures that capture the time/space costs incurred in using the system to cross channels; we go on to exhibit a system $\mathcal{G}$ that improves upon $\mathcal{F}$ with respect to these efficiency measures.

\subsection{Applications}\label{sec:conapp}

We finish by noting some potential applications of the present paper's theoretical contribution (namely, that channels can be crossed in two dimensions without disrupting their messages).
We consider, then, practical settings in which is inherent a restriction to two dimensions.

\begin{btlists}
  \item A potential application, suggested by Peter Covey-Crump, is in \emph{chip design}.
      Components on a chip are connected by conductive tracks printed onto the chip's layers.
      If two tracks cross on a single layer, then they are necessarily electrically connected; if such connection is not desired (e.g.,\ if one track is to form a connection between components $N$ and $S$, and the other, \emph{independently}, between $E$ and $W$), then one track must sidestep to another layer and bridge over the other track using `vias' (connections between corresponding points on different layers), which incurs expense (not least because of the necessity of the chip's consisting of several layers).
      We suggest that the present paper's scheme---i.e.,\ the approach of splitting each signal into two streams (which are then crossed without loss), though not necessarily implemented via cellular-automatic means---may offer scope for a preferable alternative to this costly, multi-layer bridging (though we concede that it is as yet far from clear how, on a single layer, sub-automaton $\mathcal{X}$ can be implemented efficiently).
  \item Another potential application, similar to the above in that `bridging' into a third dimension is costly, is in \emph{civil engineering}.
      Specifically, \emph{road junctions} playing the roles of $\mathcal{X}$, $\mathcal{Y}$ and $\mathcal{L}$ can be used to implement the scheme of the present paper, so as to cross two roads without the need for (and expense of) a bridge carrying one over the other.
      The form of these three junction types may for example be ($\mathcal{Y}$)~a division of alternate cars (or alternate blocks of $n$ consecutive cars) into two parallel streams (consisting, therefore, alternately of cars/blocks and gaps); ($\mathcal{L}$)~a re-merging of the streams; and ($\mathcal{X}$)~a level-crossing-style junction whereby two `car-gap-car-gap{\ldots}'\ streams are crossed, the cars of each stream synchronizing with and passing through the gaps of the other.
\end{btlists}

We defer the development of these applications and others to future work.

\section*{Acknowledgements}

This paper is based upon (and extended from) unpublished research presented in \cite{msc}.
Accordingly, we thank Richard Brent, the supervisor of that project, for his detailed comments; and Samson Abramsky and Peter Covey-Crump, for kindly taking the time to examine the project and for their insightful discussion.
We thank the three anonymous \emph{NCMA 2012} referees for their helpful suggestions and comments.
We acknowledge the generous financial support of the Leverhulme Trust, which funds the author's current position.

\end{document}